\documentclass[11pt]{article}
\usepackage{fullpage}
\usepackage{amsthm, amsmath, amssymb, graphicx}
\usepackage{algorithm}
\usepackage{algpseudocode}
\newtheorem{theorem}{Theorem}

\newtheorem{lemma}{Lemma}

\usepackage{url}
\usepackage{subfig}

\renewcommand{\a}{\alpha}
\def\algname{LSF-Join}
\def\filname{Fast-Filter}
\usepackage{booktabs}
\usepackage{multirow}

\usepackage{color,hyperref}
    \catcode`\_=11\relax
    \newcommand\email[1]{\_email #1\q_nil}
    \def\_email#1@#2\q_nil{%
      \href{mailto:#1@#2}{{\emailfont #1\emailampersat #2}}
    }
    \newcommand\emailfont{\small \sffamily}
    \newcommand\emailampersat{{\small@}}
    \catcode`\_=8\relax 

	\title{LSF-Join: Locality Sensitive Filtering for Distributed \\All-Pairs Set Similarity Under Skew}
    \author{
    Cyrus Rashtchian\\UCSD\\
    \email{crashtchian@eng.ucsd.edu}
    \and
    Aneesh Sharma\\Google\\
    \email{aneesh@google.com}
    \and 
    David P. Woodruff\\CMU\\
    \email{dwoodruf@cs.cmu.edu}
    }

\begin{document}
\maketitle
	\begin{abstract}
	All-pairs set similarity is a widely used data mining task, even for large and high-dimensional datasets. Traditionally, similarity search has focused on discovering very similar pairs, for which a variety of efficient algorithms are known. However, recent work highlights the importance of finding pairs of sets with relatively small intersection sizes. For example, in a recommender system, two users may be alike even though their interests only overlap on a small percentage of items. In such systems, some dimensions are often highly skewed because they are very popular. Together these two properties render previous approaches infeasible for large input sizes. 
	To address this problem, we present a new distributed algorithm, \algname{}, for approximate all-pairs set similarity. The core of our algorithm is a randomized selection procedure based on Locality Sensitive Filtering. 
	Our method deviates from prior approximate algorithms, which are based on Locality Sensitive Hashing.  Theoretically, we show that \algname{} efficiently finds most close pairs, even for small similarity thresholds and for skewed input sets. We prove guarantees on the communication, work, and maximum load of  \algname{}, and we also experimentally demonstrate its accuracy  on multiple graphs.
\end{abstract}

\section{Introduction}
Similarity search is a widely used primitive in data mining applications, and all-pairs similarity in particular is a common data mining operation~\cite{afrati-fuzzy, bayardo, mmds, xiao}. Motivated by recommender systems and social networks, we design algorithms for computing all-pairs set similarity (a.k.a., a set similarity join). In particular, we consider the similarity of nodes in terms of a bipartite graph. We wish to determine similar pairs of nodes from one side of the graph. For each node $v$ on the right, we consider its neighborhood $\Gamma(v)$ on the left. Equivalently, we can think of $\Gamma(v)$ as a set of the neighbors of $v$ in the graph. 
Using this representation, many graph-based similarity problems can be formulated as finding pairs of nodes with significantly overlapping neighborhoods.
We focus on the cosine similarity between pairs  $\Gamma(v)$ and $\Gamma(u)$ represented as high-dimensional vectors. 

Although set similarity search has received a lot of attention in the literature, there are three aspects of modern systems that have not been adequately addressed yet. Concretely, we aim to develop algorithms that come with provable guarantees and that handle the following three criteria: 
\begin{enumerate}
	\item {\bf Distributed and Scalable.} The algorithm should work well in a distributed environment like MapReduce, and should scale to large graphs using a large number of processors.  
	\item {\bf Low Similarity.} The algorithm should output most pairs of sets with relatively low normalized set similarity, such as a setting of cosine similarity $\tau$ taking values $0.1 \leq \tau \leq 0.5.$
	\item {\bf Extreme Skew.} The algorithm should provably work well even when the dimensions (degrees on the left) are highly irregular and skewed.
\end{enumerate}

The motivation for these criteria comes from recommender systems and social networks. For the first criteria, we consider graphs with a large number of vertices. For the second, we wish to find pairs of nodes that are semantically similar without having a large cosine value. This situation is common in collaborative filtering and user similarity~\cite{SSG17}, where two users may be alike even though they overlap on a small number of items (e.g., songs, movies, or citations).
Figure~\ref{fig:sim-pair-distribution} depicts the close pair histogram of a real graph, where most similar pairs have low cosine similarity.
For the third criteria, skewness has come to recent attention as an important property~\cite{augsten2013similarity, mccauley2018set, zhu2016lsh}, and it can be thought of as power-law type behavior for degrees on the left. 
In contrast, most other prior work assumes that the graph has uniformly small degrees on the left~\cite{mann2016empirical, SSG17, silva-survey}. This smoothness assumption 
is reasonable in settings when the graph is curated by manual actions (e.g., Twitter follow graph).  However, this 
is too restrictive in some settings, such as a graph of documents and entities, where entities can legitimately have high degrees, and throwing away these entities may remove a substantial source of information.
Another illustration of this phenomenon can be observed even on human-curated graphs, e.g., the Twitter follow graph, where computing similarities among consumers (instead of producers, as in~\cite{SSG17}) runs into a similar issue.

Previous work fails to handle all three of the above criteria.  
When finding low similarity items (e.g., cosine similarity $< 0.5$), standard techniques like Locality-Sensitive Hashing~\cite{HIM, plsh, zhu2016lsh} are no longer effective (because the number of hashing iterations is too large). Recently, there have been several proposals for addressing this, and the closest one to ours is the wedge-sampling algorithm from~\cite{SSG17}. However, the approach in~\cite{SSG17} has one severe shortcoming: it requires that each dimension has a relatively low frequency (i.e., the bipartite graph has small left degrees). 

In this work, we address this gap by presenting a new distributed algorithm \algname{} for approximate all-pairs similarity that can scale to large graphs with high skewness. As a main contribution, we provide theoretical guarantees on our algorithm, showing that it achieves very high accuracy. We also provide guarantees on the communication, work, and maximum load in a distributed environment with a very large number of processors. 

Our approach uses Locality Sensitive Filtering (LSF)~\cite{christiani2017framework}. This is a variant of the ideas used for Locality Sensitive Hashing (LSH). The main difference between LSF and LSH is that the LSF constructs a single group of surviving elements based on a hash function (for each iteration). In contrast, LSH constructs a whole hash table, each time, for a large number of iterations. While the hashing and sampling ideas are similar, the benefit of LSF is in its computation and communication costs. Specifically, our LSF scheme will have the property that if an element $v$ survives in $k'$ out of $k$ total hash functions, then the computation scales with $k'$ and not $k$. For low similarity elements, $k'$ is usually substantially smaller than $k$, resulting in a lower overall cost (for example $k'$ will be sublinear, while $k$ is linear, in the input size). We also provide an efficient way to execute this filtering step on a per-node basis. 

Our LSF procedure can also be a viewed as a pre-processing step before applying any all-pairs similarity algorithm (even one needing a smaller problem size and a graph without skew). The reason is that the survival procedure outputs a number of smaller subsets of the original dataset, each with a different, smaller set of dimensions, along with a guarantee that no dimension has a high degree. The procedure also ensures that similar pairs are preserved with high probability. Then, after performing this filtering, we may use other steps to improve the computation time. For example, applying a hashing technique may reduce the effective dimensionality without affecting the similarity structure. 

\subsubsection*{Problem Set-up}
The input consists of a bipartite graph $G$ with a set of $M$ vertices on the left and $N$ vertices on the right. We denote that graph as $G=(U,V,E)$, and we refer to $U$ as the set of {\em dimensions}, and to $V$ as the set of {\em nodes}. Given a parameter $\tau > 0$,
we want to output all similar pairs of nodes $(v,v')$ from $V$ such that 
$$
\frac{|\Gamma(v)\cap \Gamma(v')|}{\sqrt{|\Gamma(v)|\cdot|\Gamma(v')|}} \geq\tau.
$$ 
This problem also encapsulates other objectives, such as finding top-$k$ results per node. 
Note that we could equivalently identify each node $v$ with the {\em set} of its neighbors $\Gamma(v) \subseteq U$, and hence, this problem is the same as the set similarity join problem with input $\{\Gamma(v) \mid v \in V\}$ and threshold $\tau$ for cosine similarity. 
We describe our algorithm in a MapReduce-like framework, and we analyze it in the massively parallel computation model~\cite{beame2013communication, koutris2018algorithmic}, which captures the theoretical properties of MapReduce-inspired models (e.g.,~\cite{afrati-dist, karloff2010model}).
\begin{figure}
    \centering
    \includegraphics[scale=0.37]{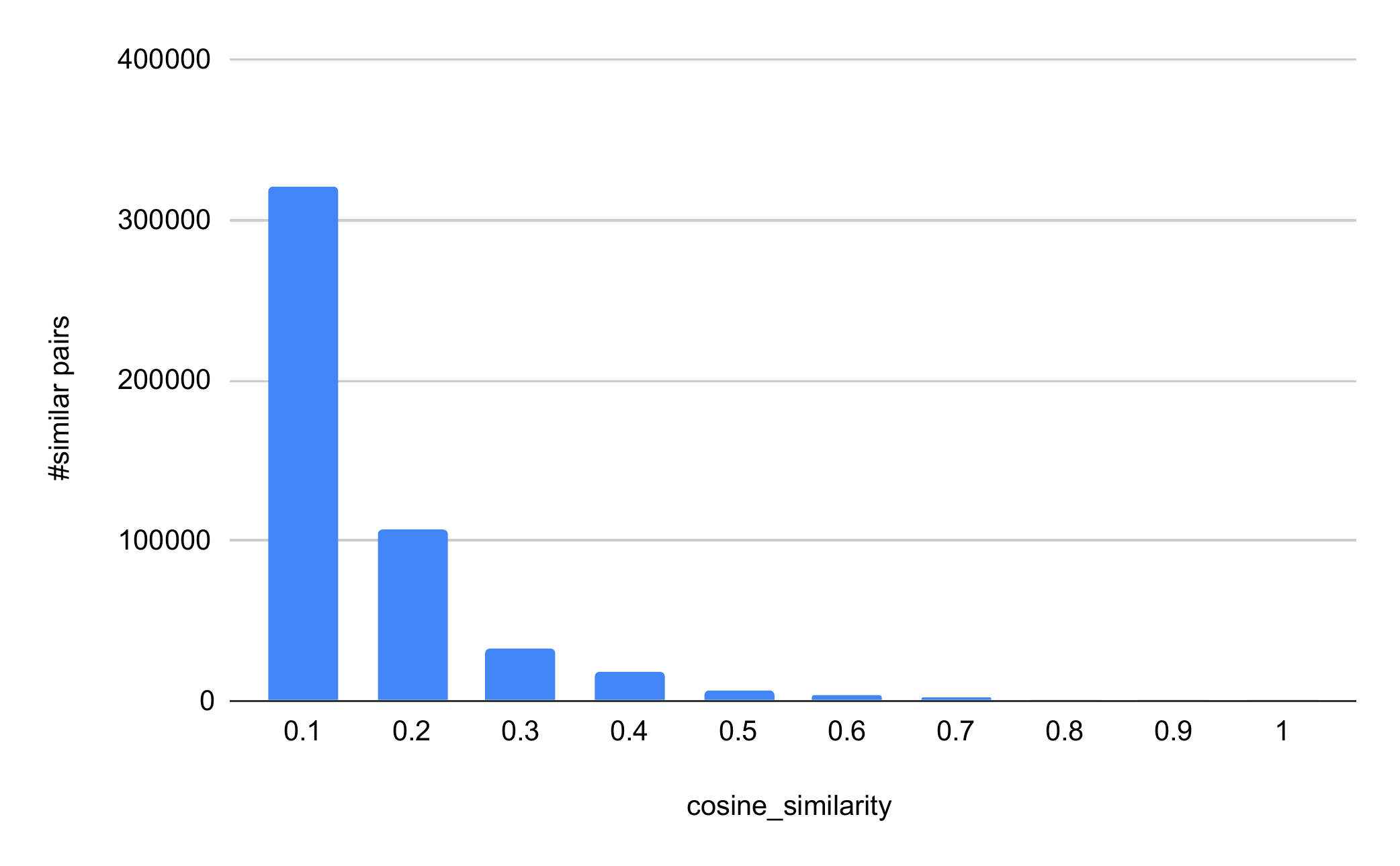}
    \caption{Histogram of the similar pairs at varying cosine similarity thresholds $\tau$ for a citation network. 
    The majority of pairs are concentrated at cosine similarity $\tau \approx 0.1$. 
    }
    \label{fig:sim-pair-distribution}
\end{figure}
We have $p$ processors, in a shared-nothing distributed environment. The input data starts arbitrarily partitioned among the processors. 
Associated to each node $v$ on the right is a vector $\Gamma(v) \in \{0,1\}^M$ which is an indicator vector for the $|\Gamma(v)|$ neighbors of $v$ on the left. 
We would like to achieve the twin properties of load-balanced servers and low communication cost. 

\subsubsection*{Our Contributions}
The main contribution of our work is a new randomized, distributed algorithm, \algname{}, which provably finds almost all pairs of sets with cosine similarity above a given threshold~$\tau$. Our algorithm will satisfy all three of the criteria mentioned above (scalability, low similarity, and skewness). A key component of \algname{} is a new randomized LSF scheme, which we call the {\em survival procedure}. The goal of this procedure is to find subsets of the dataset that are likely to contain similar pairs. In other words, it acts as a filtering step. Our LSF procedure comes with many favorable empirical and theoretical properties. First, we can execute it in nearly-linear time, which allows it to scale to very large datasets. Second, we exhibit an efficient way to implement it in a distributed setting with a large number of processors, using only a single round of communication for the whole \algname{} algorithm. Third, the survival procedure leads to sub-quadratic local work, even when the dimensions are highly skewed and the similarity threshold is relatively low.  To achieve these properties, we demonstrate how to implement the filtering using efficient, pairwise independent hash functions, and we show that even in this setting, the algorithm has good provable guarantees on the accuracy and running time. 
We also present a number of theoretical optimizations that better illuminate the behavior of the algorithm on datasets with different structural properties. Finally, we empirically validate our results by testing \algname{} on multiple graphs.


\subsubsection*{Related Work}
Many filtering-based similarity join algorithms provide exact algorithms and rely on heuristics to improve the running time~\cite{alabduljalil2013optimizing, baraglia2010document, fier2018set, mann2016empirical, vernica, wang2012can,  wang2017leveraging, xiao}. We primarily review prior work that is relevant to our setting and provides theoretical guarantees.

One related work uses LSF for set similarity search and join on skewed data~\cite{mccauley2018set}. Their {\em data dependent} method leads to a sequential algorithm based on the frequency of dimensions, improving a prior LSF-based algorithm~\cite{christiani2017framework}.
Unfortunately, it seems impossible to adapt their method to the one-round distributed setting.  
Another relevant result is the wedge-sampling approach in~\cite{SSG17}. 
They provide a distributed algorithm for low-similarity joins on large graphs. 
However, their algorithm assumes that the dataset is {\em not} skewed. 

In the massively-parallel computation model~\cite{beame2013communication, bks-skew}, multi-round algorithms have been developed that build off of LSH for approximate similarity joins, achieving output-optimal guarantees on the maximum load~\cite{hu2019output, mccauley2018adaptive}. However, it can be prohibitively expensive to use multiple rounds in modern shared-nothing clusters with a huge number of processors. 
In particular, the previous work achieves good guarantees only when the number of nodes $N$ and number of processors $p$ satisfy $N \geq p^{1+c}$ for a constant $c >0$. We focus on one-round algorithms, and we allow the possibility of $p = \Theta(N)$, which may be common in very large computing environments. 
Algorithms using LSH  work well when $\tau$ is large enough, such as $0.6 \leq \tau < 1.0$. However, for smaller $\tau$, LSH-based distributed algorithms require too much computation and/or communication due to the large number of repetitions~\cite{christiani2018scalable, lsh-practical, SSG17,yu2016generic}. Prior work has also studied finding {\em extremely} close pairs~\cite{afrati-anchor, afrati-fuzzy,beame2017massively} or finding pairs of sets with constant-size intersection~\cite{deng2018overlap}. These results do not apply to our setting because we aim to find pairs of large-cardinality sets with cosine similarity $\tau$ in the range $0.1 \leq \tau \leq 0.5$, and we allow for the intersection size to be large in magnitude. 

Finally, there are also conditional lower bounds showing that provably sub-quadratic time algorithms for all pairs set similarity (even approximate) may not exist in general~\cite{ahle2016complexity, pagh2019hardness}.

\section{The \algname{} Algorithm}
We start with a high-level overview of our set similarity join algorithm, \algname{}, which is based on a novel and effective LSF scheme. Let $G = (U,V,E)$ be the input graph with $|U|=M$ dimensions on the left, and $|V| = N$ nodes on the right. For convenience, we refer to the vertices $V$ and their indices $[N]$ interchangeably, where we use $[N]$ to denote the set $\{1,2,\ldots, N\}$.

The \algname{} algorithm uses $k$ independent repetitions of our filtering scheme (where $k \approx N$ achieves the  best tradeoff). In the $i$-th repetition we create a set $S_i \subseteq [N]$ of {\it survivors} of the set $[N]$ of vertices on the right. We will define the LSF procedure shortly, which will determine the subsets $\{S_i\}_{i=1}^k$ in a data-independent fashion. 
During the communication phase, the survival sets will be distributed in their entirety across the processors. In particular, if there are $p$ processors, then each processor will handle roughly $k/p$ different repetitions.  During the local computation, the processors will locally compute all similar pairs in $S_i$ for $i\in [k]$ and output these pairs in aggregate (in a distributed fashion). As part of the theoretical analysis, we show that the size of each $S_i$ is concentrated around its mean, and therefore, our algorithm has balanced load across the processors. To achieve high recall of similar pairs, we will need to execute the \algname{} algorithm $O(\log N)$ times independently, so that the failure probability will be polynomially small. Fortunately, this only increases the communication and computation by a $O(\log N)$ factor. We execute the iterations in parallel, and \algname{} requires only one round of communication.   

\subsection{Constructing the Survival Sets $S_i$}

We now describe our LSF scheme, which boils down to describing how to construct the $S_i$ survival sets. 
We have two main parameters of interest: $\alpha \in (0,1/2]$ denotes the survival probability of a single dimension (on the left), and $k$ denotes the number of repetitions.   The simplest way to describe our LSF survival procedure goes via uniform random sampling. We refer to this straightforward scheme as the {\em Naive-Filter} method, and we describe it first. Then, we explain how to improve this method by using a pairwise independent filtering scheme, which will be much more efficient in practice. We refer to the improved LSF scheme as the {\em \filname{}} method. Later, we also show that \filname{} enjoys many of the same theoretical guarantees of Naive-Filter, with much lower computational cost.

\paragraph{Naive-Filter.}  For the naive version of our filtering scheme, consider a repetition number $i \in [k]$. We choose a uniformly random set $U_i \subseteq U$ of vertices on left by choosing each node $u \in U$ to be in $U_i$ with probability $\a$ independently. Then, we filter vertices $v$ on the right depending on whether their neighborhood is completely contained in $U_i$ or not (that is, whether $\Gamma(v) \subseteq  U_i$ or not). The $i$-th survival set $S_i$ will be the set of vertices $v \in V$ such that $\Gamma(v) \subseteq  U_i$. We repeat this process independently for each $i=1,2,\ldots, k$, to derive $k$ filtered sets of vertices $S_1,\ldots, S_k$. Notice that for each $i$, the probability that $v$ survives in $S_i$ is exactly $\a^{|\Gamma(v)|}$, where $|\Gamma(v)|$ is the number of neighbors of $v$ on the left. 

The intuition behind using this filtering method for set similarity search is that similar pairs are relatively likely to survive in the same set. Indeed, the chance that both $v$ and $v'$ survive in $S_i$ is equal to  $\a^{|\Gamma(v) \cup \Gamma(v')|}$. When the cosine similarity is large, we must have that $|\Gamma(v) \cap \Gamma(v')|$ is large and also that $|\Gamma(v) \cup \Gamma(v')|$ is much smaller than $|\Gamma(v)| + |\Gamma(v')|$. In other words, $v$ and $v'$ are more likely to survive together if they are similar, and less likely if they are very different. For example, consider the case where $d = |\Gamma(v)| = |\Gamma(v')|$ is a large constant. Then, pairs with cosine similarity at least $\tau$ will survive together with probability $\a^{(2-\tau)d}$. At the other extreme, disjoint pairs only survive together with probability $\a^{2d}$.

The main drawback of the Naive-Filter method is that it takes too much time to determine all indices $i$ such that $v \in S_i$. Consider the set of $v$'s neighbors $\Gamma(v)$. We need to determine whether $\Gamma(v) \subseteq U_i$ for every $i \in [k]$. Hence, it requires at least $O(\alpha |\Gamma(v)| k)$ work to compute the indices where $v$ survives, that is, the set $\{i : v \in S_i \}$. We will need to set $k \gg N$, and hence, the work of  Naive-Filter is linear in $N$ or worse for each node $v$. To improve upon this, our \filname{} method will have work proportional to $|\{i : v \in S_i \}|$, and we show that this is often considerably smaller than $k$.

\subsection{The \filname{} Method} 

The key idea behind our fast filtering method is to develop a pairwise independent filtering scheme that approximates the uniform sampling of the survival sets. We then devise a way to efficiently compute the survival sets on a per-node basis, by using fast matrix operations. More precisely, for each node $v$ on the right, \filname{} will determine the indices $I_v \subseteq [k]$ of survival sets in which $v$ survives (that is, we have $I_v = \{i : v \in S_i\}$). We develop a way to compute $I_v$ independent for each vertex $v$ by using Gaussian elimination on binary matrices. The \filname{} method only requires a small amount of shared randomness between the processors.

To describe the \filname{} method, it will be convenient  to assume that $\log_2(1/\a)$ and $\log_2 k$ are both integers. We now explain the pairwise independent filtering scheme. For each node $u \in U$ on the left, we sample a random $\log_2(1/\a) \times \log_2 k$ binary matrix $A'_u$ and a $\log_2 1/\a$-length 
bit-string $b'_u$. We identify each of the $k$ repetitions $i \in [k]$ with binary vectors in the $\log_2(k)$-dimensional vector space over $GF(2)$, the finite field with two elements. In other words, we use the binary representation of $i$ to associate $i$ with a length $\log_2 k$ bit-string, and we perform matrix and vector operations modulo two.
We abuse notation and use $i$ for both the integer and the bit-string, where context will distinguish the two.

\begin{algorithm}[t]
	\caption{\ Efficient LSF for a Single Node}
	\begin{algorithmic}[1]
		\Function{\filname{}}{\ $G$, $v$,\ $\alpha$,\ $k$\ }
		\State Compute $A^v$ and $b^v$ using the shared random seed
		\State Determine the solution space of $A^vi + b^v = 0$ 
		\State Let $I_v \leftarrow \{i : A^vi + b^v = 0\}$
		\State \Return{Return $I_v$} \hspace{.8in} //\ $I_v= \{i : v \in S_i\}$
		\EndFunction
	\end{algorithmic} \label{alg:filter}
\end{algorithm}

\begin{algorithm}[t]
	\caption{\ Approximate Cosine Similarity Join}
	\begin{algorithmic}[1]
	    \State Repeat the following procedure $O(\log N)$ times in parallel:
		\Function{\algname{}}{\ $G = (U,V,E)$, \ $\tau$, \ $\a$, \ $k$\ }
		\State {\bf For} each vertex $v \in V$ do in parallel:
		\State \hspace{.15in} {\sc \filname{}}$(G,v,\a,k)$ to determine sets containing $v$
		\State Partition the sets $S_1,\ldots, S_k$ across processors
		\State Locally compute all pairs in each $S_i$ with similarity $\geq \tau$
		\State Output all close pairs in a distributed fashion
		\EndFunction
	\end{algorithmic} \label{alg:main}
\end{algorithm}

To determine whether a node $v \in [N]$ survives in $S_i$, we perform the following operation.  We first stack the matrices $A'_u$ on 
top of each other for each of $v$'s neighbors $u \in \Gamma(v)$. This forms  a $|\Gamma(v)| \cdot \log_2(1/\a) \times \log_2 k$ matrix $A^v$. We also stack 
the vectors $b'_u$ on top of each other, forming a length $|\Gamma(v)| \cdot \log_2(1/\a)$ bit-string~$b^v$. Finally, we define $S_i$ by setting 
$v \in S_i$ if and only if $A^vi+b^v = 0$, 
where $0$ denotes the all-zeros vector. 
We say that $v$ survives the $i$-th repetition if 
$A^vi+b^v = 0$. Then $I_v= \{i : v \in S_i\}$ is the set of indices $I_v \subseteq[k]$ in which $v$ survives.

In a one-round distributed setting, the processors can effectively pre-compute the submatrices  $A'_u$ and the subvectors $b'_u$ using a shared seed. In particular, these may be computed on the fly, as opposed to stored up front, by using a shared random seed and by using an efficient hash function to compute the elements of $A'_u$ and  $b'_u$ only when processing $v$ such that $u \in \Gamma(v)$. By doing so, the processors will use the same values of $A'_u$ and  $b'_u$ as one another, leading to consistent survival sets, without incurring any extra rounds of communication. 

To gain intuition about this filtering procedure, let $d = |\Gamma(v)|$ denote the number of $v$'s neighbors. Node $v$ will survive in $S_i$ if $i$ satisfies  $A^vi+b^v = 0$. This consists of $d \cdot \log_2(1/\a)$ linear equations that $i$ must satisfy. As the matrix $A^v$ and the vector $b^v$ are chosen uniformly at random, it is easy to check that $v$ survives in $S_i$ with probability $\a^{|\Gamma(v)|} = \a^d$, and hence,  their expected sizes satisfy
$$\mathbb{E}[|S_i|] = \a^d N \qquad \mbox{and} \qquad \mathbb{E}[|I_v|] = \a^d k$$ 
over a random $A^v$ and $b^v$. 

Theoretically, the main appeal of \filname{} is that it is pairwise independent in the following sense. For any two distinct repetitions $i$ and $i'$, the bit-strings for $i$ and $i'$ differ in at least one bit. Therefore, we see that $A^vi+b^v = 0$ is satisfied or not independently of $A^vi'+b^v = 0$, over the random choice of $A^v$ and $b^v$. While this is only true for pairs of repetitions, this level of independence will suffice for our theoretical analysis. Furthermore, we show that we can determine the survival sets containing $v$ in time proportional to the number $|I_v|$ of such sets, which is often much less than the total number $k$ of possible sets.

We now explain how to efficiently compute the survival sets on a per-node basis. For a fixed node $v \in [N]$, the \filname{} method determines the repetitions $i$ that $v$ survives in, or in other words, the set $I_v = \{i : v \in S_i\}$. This is equivalent to finding all length $\log_2(k)$ bit-strings $i$ that are solutions to $A^vi+b^v = 0$.
The processor can form $A^v$ and $b^v$ in $O(d)$ time, where $d = |\Gamma(v)|$,
assuming the unit cost RAM model on words of $O(\log_2(N))$ bits. 
Then, we can use Gaussian elimination over bit-strings to very quickly find all $i \in [k]$ that satisfy $A^v i + b^v = 0$. 
To understand the complexity of this, first note that $A^v$ has $\log_2 k$ columns. Moreover, without loss of generality, we see that $A^v$ has at most $\log_2 k$ rows, as otherwise there exists no solution. Therefore, 
Gaussian elimination takes $O(\log^3 k)$ time to write $A^v$ in upper triangular form (and correspondingly rewrite $b^v$) so that all solutions to $A^ui = b^u$ can be enumerated in time proportional to the number of solutions to this equation. The expected total work is $$O(N \log^3 k + \a^dkN).$$ This can be parallelized for each node $v$ independently.

We prove guarantees about \filname{} in Theorem~\ref{thm:work}. 
The pseudo-code for \filname{} appears as Algorithm~\ref{alg:filter}. 
The main difference between the two filtering methods is how the random survival sets are chosen. For the sake of this discussion, we set $k = N$, which is reasonable in practice, and we continue to let $d = |\Gamma(v)|$. In the \filname{} method, we use a random linear map
over $GF(2)$ with enough independent randomness to decide for each repetition, whether or not a node survives not. By using Gaussian elimination, we are able to compute $I_v = \{i: v \in S_i\}$ in time proportional to $|I_v| \leq N$. In particular, the amount of work for $v$ is $O(\log^3 N+ \a^dN)$ in expectation, because $\mathbb{E}[|I_v|] = \a^dN$ when $k=N$. 

	

The pseudo-code for \algname{} appears as Algorithm~\ref{alg:main}.
We assume that the vertices $v$ start partitioned arbitrarily across $p$ processors. For each vertex $v$ in parallel, we use \filname{} determine the indices $I_v$  of the sets in which $v$ survives. As detailed above, we can do so consistently by using a shared random seed for \filname{}. During the communication phase, we randomly distribute the sets $S_1,\ldots, S_k$ across $p$ processors, so that each processor handles $k/p$ sets in expectation. Then, during local computation, we compare all pairs in $S_i$ for each $i \in [k]$ in parallel. 
We use $O(\log N)$ independent iterations of the algorithm in parallel to find all close pairs with high probability (e.g., recall close to one).
Finally, we output all pairs with cosine similarity at least $\tau$ in a distributed fashion.  



One way of processing each $S_i$ set is to compare all pairs in this set. Specifically, for all pairs of nodes $v,v' \in S_i$, explicitly compute $|\Gamma(v) \cap \Gamma(v')|$ and check if it is at least $\tau \sqrt{|\Gamma(v)|\cdot|\Gamma(v')|}$. One can assume the lists $\Gamma(v)$ and $\Gamma(v')$ are sorted arrays of
$d'$ integers, where $d' = \max\{|\Gamma(v)|, |\Gamma(v')|\}$. Thus, one can compute $|\Gamma(u) \cap \Gamma(v)|$ by merging these sorted lists in $O(d')$ time,
assuming words of length $O(\log_2(N))$ can be manipulated in constant time in the unit cost RAM model. 

Letting $d_i$ be the maximum of $|\Gamma(v)|$ over $v \in S_i$, the time to locally compare all pairs in set $S_i$ is $O(|S_i|^2 d_i)$. We can also bound the average amount of work across $p$ processors to handle
all sets $S_1,\ldots, S_k$. This can be bounded by $$O\left(\sum_{i=1}^k |S_i|^2 \cdot d_i \cdot \frac{k}{p}\right).$$ 
We call this the {\it brute-force all-pairs} algorithm. 

\subsubsection{Setting the Parameters}
Let $\bar{d}$ denote the average degree on the right in the input graph. Ideally, these parameters should satisfy
\begin{eqnarray}\label{eqn:key}
\a^{(2-\tau)\bar{d}} \cdot k = 2,
\end{eqnarray} 
or in other words, $\a = (2/k)^{1/((2-\tau)\bar{d})}$, where 2 could be replaced with a larger constant for improved recall. If it is possible to approximately satisfy (\ref{eqn:key}) with $\log_2(1/\a)$ being an integer, then running $O(\log N)$ independent iterations of the algorithm with these parameters will work very well. For example, this is the case when $(1/2)^{\bar{d}} = 1/N^c$ for constant $c \approx 1$. However, for large average degree $\bar{d}$, the parameter $\a$ may exceed 1/2. To approximate $\a > 1/2$, we can subsample the matrices $A^v$ and vectors $b^v$ to increase the effective collision probability. More precisely, consider $d = |\Gamma(v)|$. If we wish to survive in a repetition with probability $\a^d$, then we can solve for $d^*$ in the equality $\a^d = (1/2)^{d^*}$, and we subsample the $d$ rows in $A^v$ and $b^v$ down to $d^*$. This effectively constructs survival sets $S_i$ as in Naive-Filter with $\a$ probability of each neighbor surviving. In the theoretical results, we will assume that $\a$ and $k$ satisfy (\ref{eqn:key}). In the experiments, we either set $\a$ to be 1/2, or we use the matrix subsampling approach; we also vary the number of independent iterations to improve recall (where we use $\beta$ to denote the number of iterations).

\section{Theoretical Guarantees}
\label{sec:theory}

We assume on the graph $G = (U,V,E)$ is right-regular with nodes in $V$ having degree $d$ for simplicity. In practice, we can repeat the algorithm for different small ranges of $d$. First, notice that
\begin{eqnarray}\label{eqn:survive}
\Pr[v \in S_i] = \Pr[A^vi + b^v = 0] = \frac{1}{2^{d \log_2 1/\a}} = \a^d
\end{eqnarray}
Now consider two nodes $u, v \in [N]$. Then both $u$ and $v$ are in $S_i$ if and only if
the following event occurs. Let $A^{u,v}$ be the matrix obtained by stacking $A^u$ on top
of $A^v$, and $b^{u,v}$ be the vector obtained by stacking $b^u$ on top of $b^v$. Note that
for each $w \in \Gamma(u) \cap \Gamma(v)$, the rows of $A_w$ occur {\it twice} in $A^{u,v}$ and the
coordinates of $b_w$ occur {\it twice} in $b^{u,v}$. Thus, it suffices to retain
only one copy of $A_w$ and $b_w$ in $A^{u,v}$ for each $w \in \Gamma(u) \cap \Gamma(v)$, and by doing
so we reduce the number of rows of $A^{u,v}$ and entries of $b^{u,v}$ to at most 
$|\Gamma(u) \cup \Gamma(v)| \cdot d \log_2 1/\a$. Consequently, 
%
\begin{eqnarray}\label{eqn:pairSurvive}
\Pr[u \in S_i \textrm{ and } v \in S_i]  = \Pr[A^{u,v} i + b^{u,v} = 0] =
 \a^{|\Gamma(u) \cup \Gamma(v)|}
\end{eqnarray}
Notice that on one extreme if $\Gamma(u)$ and $\Gamma(v)$ are disjoint, then (\ref{eqn:pairSurvive}) evaluates to $\a^{2d}$. On the other hand, if $|\Gamma(u) \cap \Gamma(v)| \geq \tau \cdot d$, then $|\Gamma(u) \cup \Gamma(v)| \leq (2-\tau)d$, and then (\ref{eqn:pairSurvive}) evaluates to $\a^{(2-\tau)d}$. \\

The discrepancy in (\ref{eqn:survive}) and (\ref{eqn:pairSurvive}) is exactly what we exploit in our LSF scheme; namely, we use the fact that similar pairs are more likely to survive together in a repetition than dissimilar pairs. 


We first justify the setting of $\alpha$ in (\ref{eqn:key}).

\begin{lemma}\label{lem:expected}
	Let $u,v$ be such that $|\Gamma(u) \cap \Gamma(v)| \geq \tau d$. The expected number of repetitions $i$ for which both $u \in S_i$ and $v \in S_i$ is at least $2$. 
\end{lemma}
\begin{proof}
	As shown in (\ref{eqn:pairSurvive}), the probability both $u$ and $v$ survive in a single repetition is 
	$\a^{|\Gamma(u) \cup \Gamma(v)|} \geq \a^{(2-\tau)d}$, and therefore the expected number of repetitions
	for which both $u \in S_i$ and $v \in S_i$ is at least $k \cdot \a^{(2-\tau)d}$, which by (\ref{eqn:key}) is
	at least $2$. 
\end{proof}

\begin{lemma}\label{lem:expected-comm}
	The expected load per processor is $\a^dNk/p$, and the expected total communication is $\a^dkN$. 
\end{lemma}
\begin{proof}
There are $k$ repetitions, each concerning one $S_i$ survival set. Each node $v \in [N]$ survives in $S_i$ with probability $\a^d$ independently. The expected size of $S_i$ is ${\bf E} |S_i| = \a^dN$. Each processor handles $k/N$ repetitions, leading to $\a^dNk/p$ expected load. The total communication is $\sum_{i=1}^k |S_i|$, which has expectation $\a^dkN.$
\end{proof}

\begin{lemma}\label{lem:expected-work}
	Using brute-force all-pairs locally, the expected work per machine is $(\a^{d}N)^2 k/p$.
\end{lemma}
\begin{proof}
	Each repetition has expected size $\a^dN$, leading to work $\a^{2d}N^2$. Each processor handles $k/p$ repetitions, implying $\a^{2d}N^2 k/p$ work per processor in expectation.
\end{proof}

Combining the lemmas and plugging in $\a$ gives us the following.

\begin{theorem}\label{thm:comm-work}
	Setting $\a = (2/k)^{1/((2-\tau)d)}$, the survival procedure has total communication is $$O(Nk^{1-1/(2-\tau)})$$ and local work 
	$$O(N^2 k^{1-2/(2-\tau)}/p)$$ in expectation.
\end{theorem}

As an example, we compare to hash-join when $p =N$, which has total communication $N^{3/2}$ and local work $N$. We set $k = N^{\frac{2-\tau}{2-2\tau}}$, and by Theorem~\ref{thm:comm-work}, the expected total communication is $Nk^{-\tau/(2-2\tau)} = N^{3/2}$. The local work per processor  is $Nk^{-\tau/(2-\tau)} = N^{1-\frac{\tau}{2-2\tau}}$. Since $\tau > 0$, the work is always sublinear, thus improving over hash-join while using the same amount of total communication.
As we will see in the theorem below,  it is crucial that we use
the family of pairwise independent hash functions above for generating our randomness. 
\begin{theorem}\label{thm:work}
	The expected total time the nodes in $[N]$ need to generate the $S_i$
	is $$O(N \log^3 k + \a^dkN + |E|),$$ and the expected total time and communication
	that the nodes in $[N]$ need to send the sets $\Gamma(v)$ for each $v \in S_i$ for each $i$
	is $$O(N \log^3 k + \a^dkN \cdot d \log N + |E|).$$
\end{theorem}
\begin{proof}
	Each node $u \in [N]$ needs
	to figure out the repetitions $i$ that it survives in. It can form $A^u$ and $b^u$ in $O(d)$ time
	assuming the unit cost RAM model on word of $O(\log_2(N))$ bits. Note $u$ then needs to figure out
	which $i \in [N]$ satisfy $A^u \cdot i + b^u = 0$. To do so, in can just solve this equation using
	Gaussian elimination. Note that $A^u$ has at most $\log_2 k$ rows, and has $\log_2 k$ columns. Therefore 
	Gaussian elimination takes at most $O(\log^3 k)$ time to write $A^u$ in upper triangular form and corresponding $b^u$ so that all solutions to the equation $A^ux = b^u$ can be enumerated in time proportional to the number of solutions to this equation. Thus, the expected time per processor
	is $O(\log^3 N + \a^dN)$, where we have used (\ref{eqn:survive}) to bound the expected number of repetitions that $u$ survives in by $k \cdot \a^d$. Thus, the total expected time to form all of the $S_i$, for $i = 1, 2, \ldots, k$, is $O(N \log^3 k + \a^dkN)$. Note that $O(\a^dkN \cdot d \log N)$ is the total expected amount of communication.
\end{proof}


While correct in expectation, since the randomness uses across the repetitions is not independent, namely,
we use the same matrices $A_w$ and vectors $b_w$ for each node $w \in [M]$, it is important to show that
the variance of the number of repetitions $i$ for which both $u \in S_i$ and $v \in S_i$ is small. This enables
one to show the probability there is at least one repetition $i$ for which both $u$ and $v$ survive is a large
enough constant, which can be amplified to any larger constant by independently repeating a constant number
of times. 

\begin{lemma}\label{lem:variance}
	Let $u,v$ be such that $|\Gamma(u) \cap \Gamma(v)| \geq \tau d$. 
	With probability at least $1/2$, there is a repetition $i$ with both $u \in S_i$ and $v \in S_i$.
\end{lemma}
\begin{proof}
	Let $X_i$ be an indicator random variable which is $1$ if $u$ and $v$ survive the $i$-th repetition, and is
	$0$ otherwise. Let $X = \sum_{i=1}^k X_i$ be the number of repetitions for which both $u$ and $v$ survive. 
	By Lemma \ref{lem:expected}, ${\bf E}[X] \geq 2$.
	It is well-known that the hash function family $f(x) = Ax + b \bmod 2$, where
	$A$ and $b$ range over all possible binary matrices and vectors, respectively, is a pairwise independent family.
	It follows that $X_1, X_2, \ldots, X_k$ are pairwise independent random variables, and consequently
	${\bf Var}[X] = \sum_{i=1}^k {\bf Var}[X_i]$. As $X_i \in \{0,1\}$, we have
	${\bf Var}[X_i] \leq {\bf E}[X_i]$, and hence, ${\bf Var}[X] \leq {\bf E}[X]$. 
	By Chebyshev's inequality,
	$$\Pr[X = 0] 
 	\ \leq\ \Pr\Big[|X - {\bf E}[X]| \geq {\bf E}[X]\Big]
\ \leq\ \frac{{\bf Var}[X]}{({\bf E}[X])^2} \ \leq\  \frac{1}{{\bf E}[X]} \ \leq\  \frac{1}{2}.$$
\end{proof}

\paragraph{Efficiently Amplifying Recall.} 
 At this point, we have shown that one iteration of \algname{} will find a constant fraction of close pairs.
To amplify the recall, we run $\beta = O(\log N)$ copies of \algname{} in parallel. We emphasize that this is a more efficient way to achieve a high probability result, better than simply increasing the number of repetitions $k$ in a single \algname{} execution. Intuitively, this is because the repetitions are only guaranteed to be pairwise independent. Theoretically, $O(\log(1/\delta))$ independent copies leads to a failure probability of $1 - \delta$ by a Chernoff bound. But, if we only increased the number of repetitions, then by Chebyshev's inequality, we would need to use $O(k/\delta)$ repetitions for the same success probability $1-\delta$. The latter requires $O(1/\delta)$ times the amount of communication/computation, while the former is only a $O(\log(1/\delta))$ factor.  Setting $\delta = 1/N^3$ leads to a failure probability of $1-1/N$ after taking a union bound over the $O(N^2)$ possible pairs. 

\section{Optimizations}
In this section, we present several extensions of the \algname{} algorithm and analysis, such as considering the number of close pairs, using hashing to reduce dimensionality, combining \algname{} with hash-join, and lowering the communication cost when the similarity graph is a matching.

\subsection{Processing Time as a Function of the Profile}
While Theorem \ref{thm:comm-work} gives us a worst-case tradeoff between computation and communication, we can better understand this tradeoff by parameterizing the total amount of work of the servers by a data-dependent quantity $\Phi$, introduced below, which may give a better overall running time in certain cases. 

Supposing that $k \geq p$, the processors receive multiple sets to process. We choose a random hash function $H : [k] \to [p]$ so that processor $j$ receives all sets $S_i$ for which $H(i) = j$. When $k \geq p$, each processor handles $k/p$ sets $S_i$ in expectation.

The processor handling the set $S_i$ receives $S_i$ together with the neighborhood $\Gamma(u)$ for each $u \in S_i$, 
and is responsible for outputting all pairs $u,v \in S_i$ for 
which $|\Gamma(u) \cap \Gamma(v)| \geq \tau d$. 

To bound the total amount of computation, we introduce a data-dependent quantity $\Phi$. Note that
the $S_i$ are independent and identically distributed, so we can fix a particular $i$. 
We define the {\it profile} $\Phi$ of a dataset as follows:
$$\Phi =  N \cdot \a^d + \sum_{u \neq v \in [N]} \a^{|\Gamma(u) \cup \Gamma(v)|}.$$

\begin{lemma}
	${\bf E}[|S_i|] = N \cdot \a^d$ and 
	${\bf E}[|S_i|^2] \leq \Phi$. 
\end{lemma}
\begin{proof}
	Let $|S_i| = \sum_u X_u$, where $X_u$ is an indicator that node $u$ survives repetition $i$. Then
	$|S_i| = \sum_u X_u$, and so ${\bf E}[|S_i|] = N \cdot \a^d$ by (\ref{eqn:survive}). 
	For the second moment, 
	$${\bf E}[|S_i|^2] = \sum_{u,v} {\bf E}[X_u X_v] \leq \sum_u {\bf E}[X_u^2] + \sum_{u \neq v}{\bf E}[X_u X_v].$$
	Plugging $\a^{|\Gamma(u)\cup \Gamma(v)|}$ from (\ref{eqn:pairSurvive}) for ${\bf E}[X_u X_v]$ and using the definition of $\Phi$ 
	proves the lemma. 
\end{proof}

We are interested in bounding the overall time for all nodes in $[p]$ to process the sets $S_i$. 
\begin{theorem}
	The total work of the nodes in $[p]$ to process the sets $S_1,\ldots, S_k$, assuming that we use the brute-force all-pairs algorithm 
	is $O(dk \Phi).$ The average work per processor is $O(\frac{dk}{p} \Phi).$
\end{theorem}
\begin{proof}
	After receiving the $S_i$, the total time for all processors to execute
	their {\it brute-force all-pairs} algorithm is $O(dk \Phi)$, which allows for
	outputting the similar pairs. The theorem follows. 
\end{proof}

\subsection{Hashing to Speed Up Processing}\label{sec:hash1}
Recall that the processor responsible for finding all similar pairs in $S_i$ receives the set $\Gamma(u)$ of neighbors of each node $u \in S_i$.
In the case when the neighborhoods are all of comparable size, sat size $d$, we can think of $\Gamma(u)$ as a vector $\chi_u \in \{0,1\}^{M}$ with exactly $d$ ones in it; here 
$\chi_u$ is the characteristic vector of the neighbors of $u$. We can first hash the vector
$\chi_u$ down to $s = d/(z \tau)$ dimensions, for a parameter $z > 0$. To do this, we
use the CountMin map \cite{cm05}, which can be viewed as a random matrix 
$S \in \{0,1\}^{s \times M}$ with a single non-zero per column, and this non-zero is chosen
uniformly at random and independently for each of the $M$ columns of $S$. We replace
$\chi_u$ with $S \cdot \chi_u$. If an entry of $S \cdot \chi_u$ is larger than $1$, we replace
it with $1$, and let the resulting vector be denoted $\gamma_u$, which is in $\{0,1\}^s$. 
Note that we can compute all of the $\gamma_u$ for a given repetition $i$
using $O(|S_i| d)$ time, assuming arithmetic operations on $O(\log M)$ bit words can be performed
in constant time. 

While $\langle \chi_u, \chi_v \rangle = |\Gamma(u) \cap \Gamma(v)|$ for two nodes
$u, v \in [N]$, it could be that $\langle \gamma_u, \gamma_v \rangle \neq |\Gamma(u) \cap \Gamma(v)|$. We
now quantify this.

\begin{lemma}\label{lem:hashing}
	For any two nodes $u,v \in [N]$, it holds that with probability at least $1-2/z$, 
	$$|\langle \gamma_u, \gamma_v \rangle - \langle \chi_u, \chi_v \rangle| \leq d\tau/2.$$
\end{lemma}
\begin{proof}
	Note that $\langle \gamma_u, \gamma_v \rangle \leq |\Gamma(u) \cap \Gamma(v)|$ since each node $w \in \Gamma(u) \cap \Gamma(v)$ is hashed to a bucket by CountMin, which will be a coordinate that is set to $1$ in both $\gamma(u)$ and $\gamma(v)$. Also the probability that $w$ hashes to a bucket containing a $\hat{w} \in \Gamma(u) \cap \Gamma(v)$ with $\hat{w}\neq w$ is at most $d/(d/(z \tau)) = z \tau,$ and the expected number of $w$ with this property is at most $d z \tau$. By a Markov bound, the number of such $w$ is at most $d\tau/2$ with probability at least $1-2/z$, as desired.
\end{proof}
By the previous lemma we can replace the original dimension-$M$ vectors $\chi_u$ with the potentially much smaller dimension-$d/(z \tau)$-vectors $\gamma_u$ with a small price in accuracy and success probability.

\subsection{Combining \algname{} with Hash-Join}

The LSH-based approach of Hu et. al~\cite{hu2019output} suggests (in our framework) an alternate strategy of sub-partitioning the survival sets, using a hash-join to distribute the brute-force all-pairs algorithm. Here we analyze this combined approach and plot the tradeoffs.  We show that this strategy does not provide any benefit in the communication vs. computation tradeoff, perhaps surprisingly. 

The combined strategy, using $p$ processors, starts by using $k = p^c$ repetitions for a parameter $c<1$, and this is followed by a hash join on each survival set. More precisely, we first construct $k$ sets $S_1,\ldots,S_k$ using the \filname{} survival procedure. Then, for each set $S_j$, we will process all pairs in $S_j \times S_j$ using $p/k = p^{1-c}$ machines. This can be implemented in one round, because all we need to do is estimate the size of each set $S_j$ approximately, that is, $|S_j| \approx \a^dN$. Then, we can implement the hash-join in a distributed fashion.

We first review the guarantees of the standard hash-join.

\begin{lemma}\label{lem:hash-join}
	For $N$ vectors and $p$ machines, a hash-join has expected total communication $N\sqrt{p}$  and expected $N^2/p$ work per machine.
\end{lemma}

We use this bound to compute the communication and work, when using a hash-join to process each survival set.

\begin{theorem}
	The combined approach has  expected total communication  $N^{1 +\frac{c(1-\tau)}{2-\tau}} p^{\frac{1- c}{2}}$ and  expected  $N^2 /p^{1+\frac{c\tau}{2-\tau}}$ work per processor.
\end{theorem}
\begin{proof}
	When $k = N^c$, we have that $\a^d = k^{\frac{-c}{2-\tau}}$, and hence, we have $|S_j| = \a^dN  = N p^{-\frac{c}{2-\tau}}$ in expectation. We use Lemma~\ref{lem:hash-join} to analyze the hash-join for each of the $p^c$ groups of $p^{1-c}$ processors. Each group handles $N' = Np^{-\frac{c}{2-\tau}}$ inputs, and therefore the communication of the group is $N'\cdot p^{1/2 - c/2}$, which is  $N^{1-  \frac{c}{2-\tau}} p^{\frac{1- c}{2}}$. Multiplying by $N^c$, the exponent of $N$ becomes 
	$$1 + c -\frac{c}{2-\tau} = 1 +\frac{c(1-\tau)}{2-\tau},$$ which gives the claimed communication bound. 
	For the per processor work, we have that this is the claimed bound:
	$$(N')^2 / p^{1-c} =  N^2 p^{-1+c-\frac{2c}{2-\tau}} = N^2 p^{-1-\frac{c\tau}{2-\tau}}.$$
\end{proof}

\begin{figure}[t]
    \centering
    \includegraphics[scale=0.65]{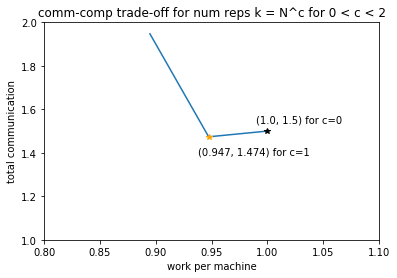}
    \caption{For $\tau =0.1$ and $p=N$, a comparison of \algname{} and the combined approach of using hash-join to distribute the brute-force all-pairs. We plot the exponent of $N$ for the different settings of $k = N^c$ repetitions for $0 < c < 2$.}
    \label{fig:comb-tradeoff}
\end{figure}
Figure~\ref{fig:comb-tradeoff} demonstrates that the combination approach is never better than the original \algname{} approach.
For a comparison, we consider $p = N$ processors, and hence, the number of repetitions will be $k = N^c$ for $0 < c < 2$. 
Then, when $c \geq 1$, the survival procedure  has  expected total communication $N^{1 + \frac{c(1-\tau)}{2-\tau}}$, and it has expected  $N^{1 -\frac{c\tau}{2-\tau}}$ work per processor.
And, when $c \leq 1$ we have that the combined approach has  expected total communication $N^{\frac{3}{2} -\frac{c \tau}{2(2-\tau)}}$, and it has expected  $N^{1 -\frac{c\tau}{2-\tau}}$ work per processor. Notice that $c = 1$ corresponds to standard \algname{}, and $c = 0$ corresponds to using a hash-join on the whole dataset.

\subsection{When the Similarity Graph is a Matching}

Recall that to recover all close pairs with high probability, we need to iterate the \algname{} algorithm $O(\log N)$ times, because each time finds a constant fraction of close pairs. We exhibit an improvement using multiple communication steps when the similar pairs are structured. An important application of all-pairs similarity is constructing the similarity graph. In our setting, the {\em similarity graph} connects all pairs $v,v' \in V$ such that their cosine similarity is at least $\tau$. The structure that we consider is when the similarity graph happens to be a {\em matching}, containing exactly $N/2$ disjoint pairs $v,v'$ with similarity at least $\tau$. 

The key idea is that each iteration decreases the number of input nodes by a constant fraction. We will remove these nodes (or at least one endpoint from each close pair) from consideration, and then repeat the procedure using the remaining nodes. We observe that this method can also be extended to near-matchings (e.g., small disjoint cliques). Similarly, our result is not specific to \algname{}, and the technique would work for any LSF similarity join method. 

We state our result using the $r$-th iterated log function $\log^{(r)} N$, where $\log^{(1)} N = \log N$, and $\log^{(r)} N = \log(\log^{(r-1)} N)$  for $r\geq 2$. Then, we show:

\begin{theorem}
Using $2r-1$ communication steps, we can find all but a negligible fraction of close pairs when the similarity graph is a matching. The
total communication and computation is $O(\log^{(r)} N)$  times the cost of one execution of \algname{}.
\end{theorem} 
\begin{proof}
For $r=1$, we simply run \algname{} $O(\log N)$ times independently in a single communication step, where each time finds a constant fraction of close pairs. For $2r-1 \geq 3$ communication steps, we will use $r$ rounds of \algname{}, and we will remove all found pairs between subsequent rounds (each round will take two communication steps, except for the last, which takes one). 

In the first round, we run \algname{} $T_r = O(\log^{(r)} N)$ times. Then, the expected number of pairs that are {\em not} found will be $O(N/2^{T_r})$, where $2^{T_r} = \mathrm{poly}(\log^{(r-1)}N)$. In the next round, with $r-1$ rounds remaining, we will only consider the remaining pairs, and we will iterate \algname{} $T_{r-1}$ times. We repeat this process until no more rounds remain, and output the close pairs from all rounds.

We can implement each round of the above algorithm using at most two communication steps. We do so by marking the found pairs between rounds using a single extra communication step. More formally,  the input pairs start partitioned across $p$ processors. We denote the input partition as $V = V_1 \cup \cdots \cup V_p$. After finding some fraction of close pairs, processor $i$ must be notified of which nodes in $V_i$ are no longer active. Whenever processor $j$ finds a close pair $(v,v')$, it sends the index of $v$  to processor $i$ such that $v \in V_i$ (and similarly for $v' \in V_{i'}$), where $i$ is known to processor~$j$ because processor $i$ must have sent $v$ to processor $j$ in  \algname{}. We reduce the total input set from $V$ to $V'$, where $V'$ denotes the remaining nodes after removing the found pairs.

To analyze this procedure, notice that the dominant contribution to the total communication and computation is the first round. This is because the subsequent rounds have a geometrically decreasing number of input nodes. The first round uses $T_r = O(\log^{(r)} N)$ iterations of \algname{}, which shows that overall communication and computation is $O(\log^{(r)} N)$  times the cost of one iteration.
\end{proof}


\subsection{Hashing to Improve Recall}
Not only is hashing helpful in order to reduce the description size of the neighborhood sets, as described in Section \ref{sec:hash1}, hashing can also be used to increase the number of similar pairs surviving a repetition, and thus the recall. Before, a node pair $(u,v)$ survives a repetition with probability $\alpha^{|\Gamma(u) \cup \Gamma(v)|}$. Hashing can, however, make $|\Gamma(u) \cup \Gamma(v)|$ smaller due to collisions. Suppose we hash the characteristic vector $\chi_u \in \{0,1\}^M$ of the neighborhood of a node $u$ down to $d/C$ dimensions for some parameter $C \geq 1$, obtaining the vector $\gamma_u \in \{0,1\}^{d/C}$, as in Section \ref{sec:hash1}. 
We could, for example, set $C = z\tau$ as in Section \ref{sec:hash1}. 

\begin{lemma}\label{lem:balls}
Thinking of $\gamma_u$ and $\gamma_v$ as characteristic vectors of sets, and letting $t = |\Gamma(u) \cup \Gamma(v)|$, we have $${\bf E}[|\gamma_u \cup \gamma_v|] = (d/C)(1-(1-C/d)^t) < t.$$
\end{lemma}
\begin{proof}
Let $X_i = 1$ be an indicator random variable for the event that $i$-th bin is non-empty when throwing $t$ balls into $d/C$ bins. If the bin is empty, then let $X_i = 0$. Then ${\bf E}[X_i] = 1-(1-C/d)^t$, and so ${\bf E}[X] = d/C - (d/C)(1-C/d)^t = d/C (1-(1-C/d)^t) \leq d/C,$ where $X = |\gamma_u \cup \gamma_v|$ is the total number of non-empty bins. 
\end{proof}
By Lemma \ref{lem:balls}, the expected size of the union of the neighborhoods drops after hashing. This is useful, as the survival probability of the node pair $(u,v)$ in a repetition after hashing is now $\alpha^{|\gamma(u) \cup \gamma(v)|}$, which by the previous lemma is larger than before since $|\gamma(u) \cup \gamma(v)| \leq |\Gamma(u) \cup \Gamma(v)|,$ and this inequality is strict in expectation. Note, however, that the communication and work per machine increase in expectation, but this tradeoff may be beneficial. 

\section{Experimental Results}
\label{sec:experiments}

In this section, we complement the theoretical analysis presented
earlier with experiments that measure the recall and efficiency of \algname{} on 
three real world graphs from the SNAP repository~\cite{snapnets}: WikiVote, PhysicsCitation, and Epinions. In accordance with our motivation, we also run \algname{} on an extremely skewed synthetic graph, on which the WHIMP algorithm fails.

\begin{table*}[t]
		\centering\small
		\setlength{\tabcolsep}{9pt}
		\renewcommand{\arraystretch}{1.5}
			\begin{tabular}{lrr|rc|cc}
				\toprule
				\multirow{2}{*}{Dataset} & \multirow{2}{*}{$N$} & \multirow{2}{*}{$M$} & \multicolumn{2}{c}{Communication Cost} \vline  & \multicolumn{2}{c}{Recall} \\
				\cline{4-7}
				& & & \algname{} & $\textrm{WHIMP}^\dagger$  & \algname{} & WHIMP\\
				\midrule
				WikiVote & 7K & 104K & 710MB ($\sum_i|S_i|=71M$, $\beta=30$) & 60MB & 100\% & 100\%\\
				Citation & 34K & 421K & 410MB ($\sum_i|S_i|=41M$, $\beta=1$) & 50MB & 100\% & 100\% \\
				Epinions & 60K & 500K & 6GB ($\sum_i|S_i|=573M$, $\beta=1$) & 60MB & 100\% & 100\% \\
				Synthetic & 10M & 200M & 160GB ($\sum_i|S_i|=8B$, $\beta=50$) & Failed & 90\% &  --- \\
				\bottomrule
				\end{tabular}
				\caption{Summary of the performance of \algname{} and WHIMP on the four datasets, in terms of communication cost and recall (precision for WHIMP was also high). We note that \algname{} was run at the minimum number of independent iterations $\beta$ to achieve high recall for $\tau = 0.1$.\\
				* The communication cost of \algname{} depends on the number of survivors, which we note along with the value of $\beta$.\\
				$\dagger$ WHIMP communication cost is dominated by shuffling SimHash sketches. We use 8K bits for SimHash, as suggested in~\cite{SSG17}.}
			    \label{tab:datasets}
\end{table*}

\subsubsection*{Experimental Setup} We compare \algname{} against
the state of the art WHIMP algorithm from~\cite{SSG17}, and hence our
setup is close to the one for WHIMP. In this vein, we transform our graphs
into bipartite graphs, either by orienting edges from left to right
(for directed graphs), or by duplicating nodes on either side (for
undirected ones). This is in accordance with the setup of the left side denoting sets and the right side denoting nodes that is described in the introduction. Also, we pre-filter each bipartite graph to have a
narrow degree range on the right (the left degrees can still be
$O(n)$) to minimize variance in cosine similarity values
due to degree mismatch. This makes the experiments cleaner, and the algorithm itself can run over all degrees
in a doubling manner. We use sparse matrix multiplication for
computing all-pairs similarity after computing the survivor sets $S_i$ for each bucket $i$, as it is quite fast in practice and consumes $d\cdot
O(|S_i|)$ memory on each server. Finally, even though we
computed a theoretically optimal value of $\alpha$ earlier, in
practice, a smaller choice of $\alpha$ often suffices in combination
with repeating the \filname{} method for $\beta \geq 1$ independent
iterations. 

For each of the graphs, we run \algname{} on the graph on a distributed MapReduce platform internal to Google, and compare the
output similar pairs against a ground truth set generated from a
sample of the data. The ground truth set is generated by doing an
exact all-pairs computation for a small subset of nodes chosen at
random. Using this ground truth, we can measure the efficacy of the
algorithm, and the measure we focus on for the evaluation is the
recall of similar pairs\footnote{The precision is dependent on the
method used to compute all-pairs similarity in a bucket, and since we
use sparse matrix multiplication, for us this is
100\%.}. Specifically, let the set of true similar pairs in the ground
truth with similarity at least $\tau$ be denoted by $S$. Furthermore,
let the set of similar pairs on the same set of nodes that are
returned by the algorithm be $\hat{S}$. Then, the recall $R =
\frac{|\hat{S} \cap S|}{|S|}$. For a fixed value of $\tau$, we can
measure the change in recall as the number $\beta$ of independent
iterations varies (with fixed $\a$ and $k =N$). We run our experiments
at a value of $\beta$ that achieves high recall (which is a strategy that carries across datasets), and the results are summarized in Table~\ref{tab:datasets} for ease of comparison. There is a synthetic dataset included in the table, which is described later. The communication cost for \algname{} is dependent on the number of survivors, which in turn depends on the choice of $\beta$. We do ignore a subtlety here in that the communication cost will actually often be much less than the number of survivors, since multiple independent repetitions will produce many copies of the same node and hence we can only send one of those copies to a processor. 

We reiterate that our experimental comparison is only against the WHIMP algorithm as the WHIMP paper demonstrated that commonly used LSH-based techniques are provably worse. Since WHIMP is only applicable in the scenario where there are no high degree left nodes, our three public graphs are those for which this assumption holds in order to be able to do a comparison. Since the WHIMP algorithm has output-optimal communication complexity, we expect WHIMP to have lower communication cost than \algname{}, as WHIMP's communication cost is dominated by the number of edges in the graph. This is indeed seen to be the case from Table~\ref{tab:datasets}. However, \algname{} trades-off higher communication cost with the benefit of load balancing across individual servers. WHIMP does not do any load balancing in the worst case, which can render it inapplicable for a broad class of graphs, as we shall see in the next section. Indeed, the WHIMP job failed for our synthetic graph. 

\begin{figure}[t]
	\centering
	\subfloat[][Recall at $\tau=0.1$]{
		\includegraphics[width=0.45\textwidth]{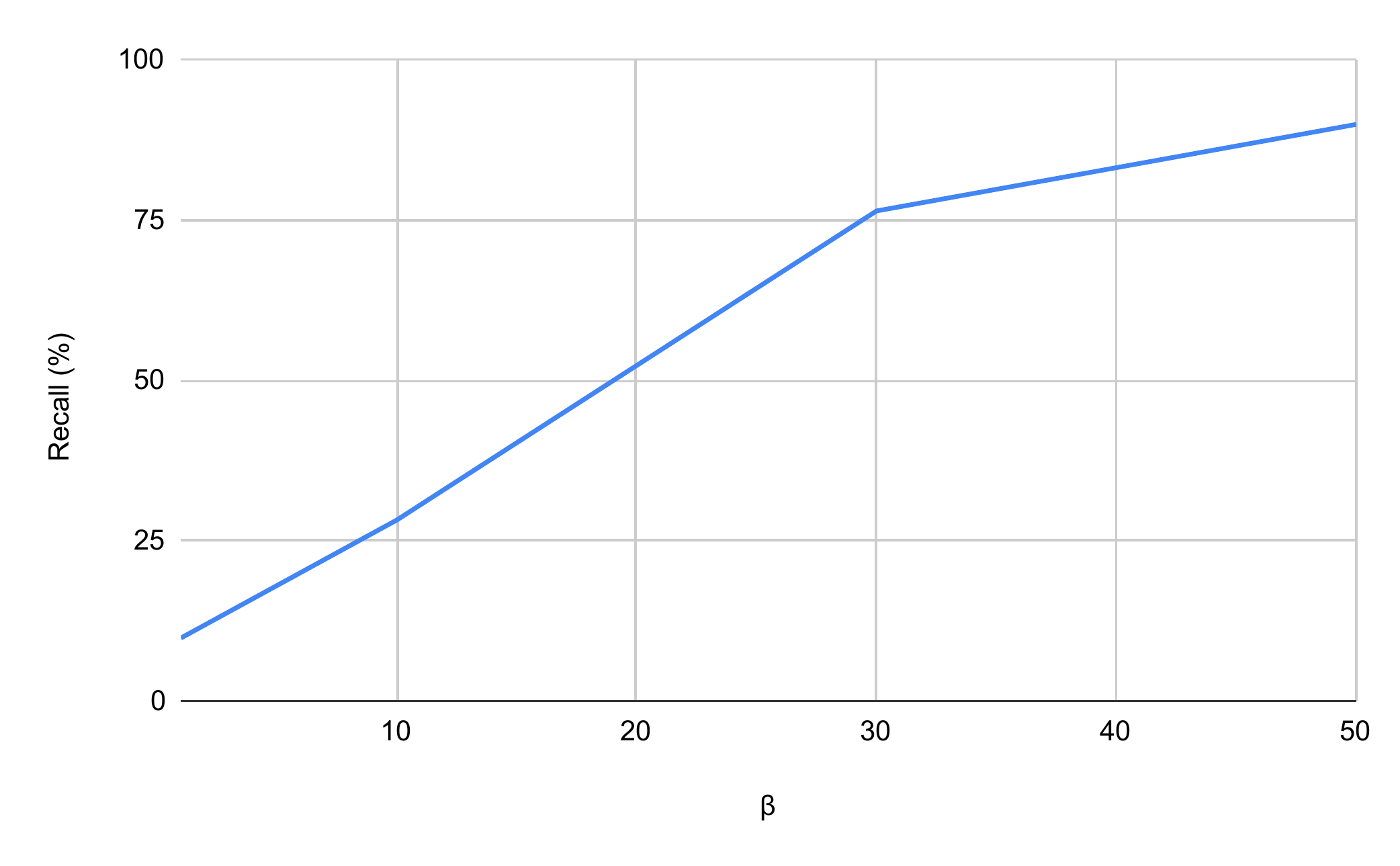}
		\label{fig:recall_random}}
	\subfloat[][Number of Survivors]{
		\includegraphics[width=0.45\textwidth]{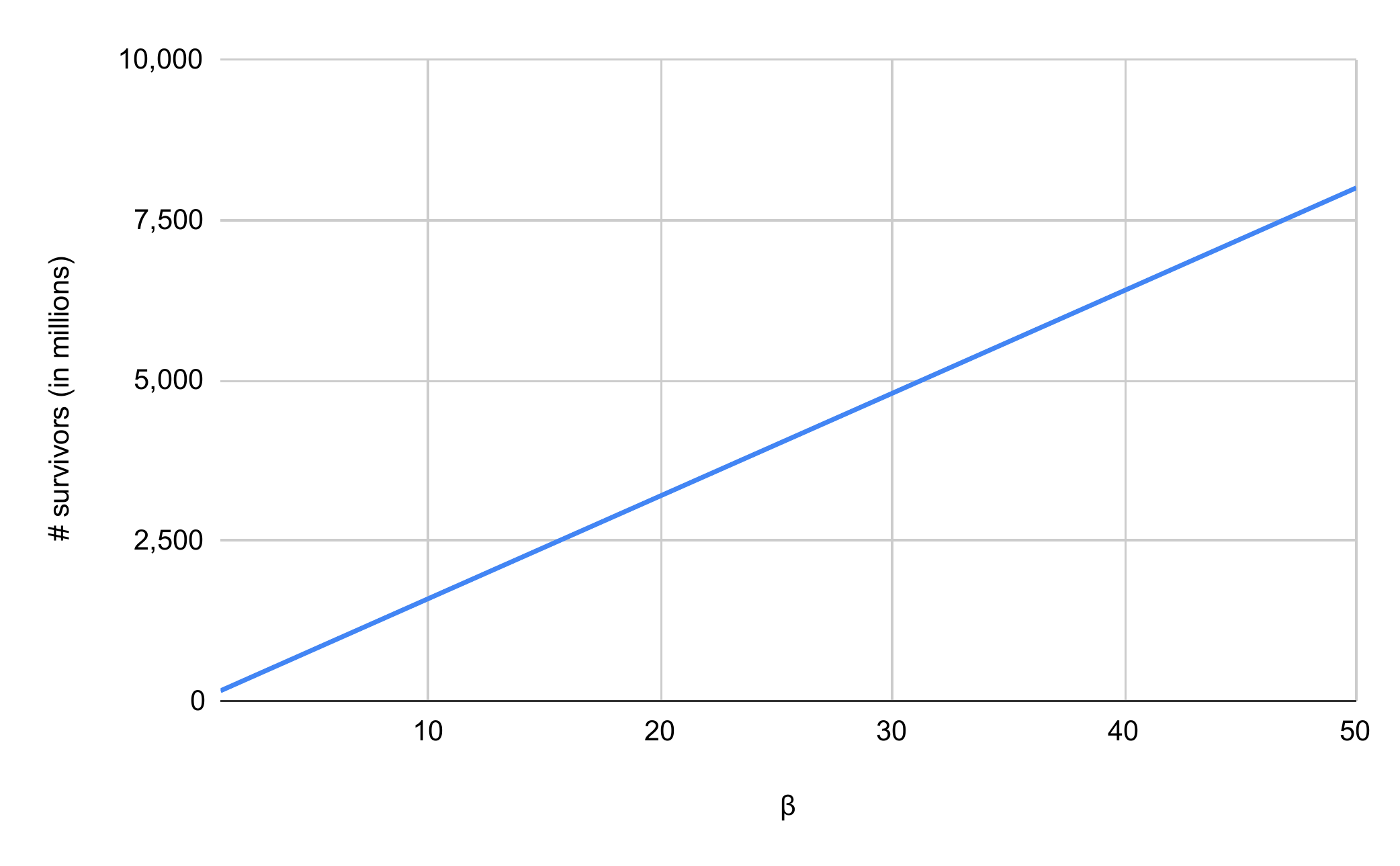}
		\label{fig:survivors_random}} 
	\caption{Recall and number of survivors as $\beta$ increases for the synthetic skewed graph.}
	\label{fig:globfig}
\end{figure}

\subsection{Synthetic Graph With Extreme Skew} \label{sec:synthetic}
To illustrate a case that WHIMP fails to address, we present results
on a synthetic graph that contains the core element of skeweness that
we set out to address in this work. We anticipate that the same results will hold for several real world settings, but a synthetic graph is sufficient for comparison with WHIMP. Indeed, the motivation for this randomly
generated synthetic graph comes from user behavior where even though
users consume almost the same amount of content (say, videos) online,
the content being consumed sees a power-law distribution (e.g., some
videos are vastly more popular than others). A simplified setting of
the same phenomenon can be captured in the following random bipartite
graph construction: we build an $N \times N$ bipartite graph
$G(U,V,E)$, where each right node has degree $d$. Each right node
$v\in V$ chooses to connect to left nodes as follows: first pick $d/2$
nodes at random (without replacement) from a small set of {\em hot}
nodes $H\subset U$, and pick $d/2$ nodes at random (again, without
replacement) from the rest of $U\setminus H$. If $|H| = \gamma\cdot
d$, and $|H|\ll N$, this results in right nodes having pairwise cosine
similarity that scale with $1/\gamma$ while the hot dimensions have
degree $O(n)$ for constant $\gamma$. In this setting, we expect wedge
sampling-based methods to fail since the hot dimensions have large
neighborhoods.

We constructed such a synthetic random bipartite graph with the
following parameters: $N=10 \textrm{ million}$, $d=20$, and $\gamma =
10$. Then, we repeated the same experiment as the one described above
for the real world graphs. This time, we noted that WHIMP
failed as the maximum degree for left nodes was around $500K$. We were
able to run our procedure though, and the recall and the communication cost of the \filname{} procedure is shown in
Table~\ref{tab:datasets}. 
The recall of the \filname{} procedure is shown in Fig~\ref{fig:recall_random}, and the number of survivors in Fig~\ref{fig:survivors_random}. Note that, as before, we are able to achieve high recall even on this graph with a heavily skewed degree distribution, with reasonable communication cost.



\section{Conclusion}
We present a new distributed algorithm, \algname{}, for approximate all-pairs set similarity search. The key idea of the algorithm is the use of a novel LSF scheme. We  exhibit an efficient version of this scheme that runs in nearly linear time, utilizing pairwise independent hash functions. We show that \algname{} effectively finds low similarity pairs in high-dimensional datasets with extreme skew. Theoretically, we provide guarantees on the accuracy, communication, and work of \algname{}. Our algorithm improves over hash-join and LSH-based methods. Experimentally, we show that \algname{} achieves high accuracy on real and synthetic graphs, even for a low similarity threshold. Moreover, our algorithm succeeds for a graph with extreme skew, whereas prior approaches fail.

\paragraph{Acknowledgments.} 
Part of this work was done while D. Woodruff was visiting Google Mountain View. D. Woodruff
also acknowledges support from the National 
Science Foundation
Grant
No. CCF-1815840.

\bibliographystyle{plain}
\bibliography{references}

\end{document}